\documentclass[11pt]{article}
\usepackage[left=1in,top=1in,right=1in,bottom=1in,head=.1in,nofoot]{geometry}

\usepackage{setspace,url,bm,amsmath} 

\usepackage{titlesec} 
\titlelabel{\thetitle.\quad} 
\titleformat*{\section}{\bf\large\center\uppercase} 

\usepackage{graphicx} 
\usepackage{latexsym}
\usepackage{authblk}
\usepackage{hyperref}

\usepackage{amsthm}
\theoremstyle{definition}

\newtheorem{theorem}{Theorem}

\newtheorem{corollary}{Corollary}

\usepackage{natbib} 

\usepackage{etoolbox} 
\apptocmd{\sloppy}{\hbadness 10000\relax}{}{} 

\begin{document}
\doublespacing
\title{\bf On finite-population Bayesian inferences for $2^K$ factorial designs with binary outcomes}
\author[1]{Jiannan Lu\thanks{Address for correspondence: Jiannan Lu, One Microsoft Way, Redmond, Washington 98052-6399, U.S.A. Email: \texttt{jiannl@microsoft.com}}}
\affil[1]{Analysis and Experimentation, Microsoft Corporation}
\date{\today}
\maketitle
\begin{abstract}
Inspired by the pioneering work of \cite{Rubin:1978}, we employ the potential outcomes framework to develop a finite-population Bayesian causal inference framework for randomized controlled $2^K$ factorial designs with binary outcomes, which are common in medical research. As demonstrated by simulated and empirical examples, the proposed framework corrects the well-known variance over-estimation issue of the classic ``Neymanian'' inference framework, under various settings.
\end{abstract}

\section{Introduction}
\label{sec:intro}

Medical researchers \citep[e.g.][]{Chalmers:1955, Hennekens:1985, Stampfer:1985, Eisenhauer:1994, Post:1997, Rapola:1997, Franke:2000, Ayles:2008, Greimel:2011, Manson:2012, James:2013} have a long history of employing randomized controlled $2^K$ factorial designs to simultaneously evaluate the causal effects of multiple two-level treatment factors on binary outcomes. To conduct causal inference on randomized controlled $2^K$ factorial designs, \cite{Dasgupta:2015} proposed a randomization-based framework based on potential outcomes \citep{Neyman:1923, Rubin:1974, Rubin:1990}. Unlike model-based approaches \citep[e.g.,][]{Simon:1997}, the ``Neymanian'' causal inference framework proposed by \cite{Dasgupta:2015} relies only on the randomization of the treatment assignment, which is often considered the ``gold standard for causal inference'' \citep{Rubin:2008}. The Neymanian framework possesses some conceptual, theoretical and practical appeals. For example, as pointed out by several researchers \citep[e.g.,][]{Miller:2006}, in some randomized experiments the participants are not a random sample from a hypothetical super-population. In such cases, finite-population analyses by the Neymanian framework might be more interpretable.

Despite the aforementioned advantages of the Neymanian causal inference framework, a long-standing challenge it faces is the over-estimation of the sampling variance of the randomization-based causal estimate, as mentioned by \cite{Aronow:2014}. A possible solution of this challenge is the finite-population Bayesian inference framework by \cite{Rubin:1978}, which uniquely combined the strengths of both the classic Neymanian and the classic Bayesian methodologies, by assuming that the potential outcomes are sampled from a hypothetical super-population, while retaining the finite-population causal effects as the inferential end-points. Realizing this salient feature, in an illuminating paper \cite{Ding:2016} developed a finite-population Bayesian framework to analyze completely randomized treatment-control studies (i.e., $2^1$ factorial designs) with binary outcomes, and showed that it indeed dominated the classic Neymanian approach. Inspired by their work, in this paper we extend \cite{Ding:2016}'s finite-population Bayesian framework to general $2^K$ factorial designs.

The remainder of the paper is organized as follows. Section \ref{sec:review} reviews the Neymanian inference framework for $2^K$ factorial designs with binary outcomes. Section \ref{sec:bayes} developed a finite-population Bayesian inference framework for $2^K$ factorial designs, by first proposing an imputation model under the assumption of independent potential outcomes, and then conducting sensitivity analysis for when the independence assumption is violated. Sections \ref{sec:simu} and \ref{sec:example} presented several simulated and empirical examples to demonstrate the proposed Bayesian methodology. Section \ref{sec:conclusion} concludes and discusses future directions. We relegate all proofs and other technial details to the Appendix.

\section{Neymanian inference}
\label{sec:review}

\subsection{Factorial designs}

In order to review the Neymanian causal inference framework for $2^K$ factorial designs, we adapt some materials from \cite{Lu:2016a, Lu:2016b}. $2^K$ factorial designs generally consist of $K$ distinct treatment factors with two-levels -1 and 1, resulting $J=2^K$ treatment combinations
$
\bm z_1, \ldots, \bm z_J.
$
To define them, we construct the $J \times J$ model matrix $\bm H = (\bm h_0, \ldots, \bm h_{J-1})$ as follows \citep[c.f.][]{Wu:2009}. First, let $\bm h_0 = \bm 1_J.$ Second, for $k=1,\ldots,K$, construct $\bm h_k$ by letting its first $2^{K-k}$ entries be -1, the next $2^{K-k}$ entries be 1, and repeating $2^{k-1}$ times. Third, if $K \ge 2,$ order all subsets of $\{1, \ldots, K\}$ with at least two elements, first by cardinality and then lexicography. For $k^\prime = 1, \ldots J-1-K,$ let $\sigma_{k^\prime}$ be the $k^\prime$th subset and $\bm h_{K+k^\prime} = \prod_{l \in \sigma_{k^\prime}} \bm h_l,$ where ``$\prod$'' stands for entry-wise product.

The $j$th row of the sub-matrix $(\bm h_1, \ldots, \bm h_K)$ is $\bm z_j,$ for $j = 1, \ldots, J.$ For example, for $K=2,$
\begin{equation*}
\bm H =
\bordermatrix{& \bm h_0 & \bm h_1& \bm h_2 & \bm h_3\cr
              & 1  & -1 &  -1  & 1 \cr
              & 1 & -1 &  1  & -1 \cr
              & 1  & 1 &  -1  & -1 \cr
              & 1 & 1 & 1 & 1},
\end{equation*}
and the treatment combinations are 
$
\bm z_1=(-1, -1),
$ 
$\bm z_2=(-1, 1),
$ 
$
\bm z_3=(1, -1)
$ 
and 
$
\bm z_4=(1, 1),
$
respectively. For $K=3,$
\begin{equation*}
\bm H =
\bordermatrix{& \bm h_0 & \bm h_1& \bm h_2 & \bm h_3 & \bm h_4 & \bm h_5& \bm h_6 & \bm h_7 \cr
              & 1 & -1 & -1 & -1 &  1 &  1 &  1 & -1 \cr
              & 1 & -1 & -1 &  1 &  1 & -1 & -1 &  1 \cr
              & 1 & -1 &  1 & -1 & -1 &  1 & -1 &  1 \cr
              & 1 & -1 &  1 &  1 & -1 & -1 &  1 & -1 \cr
              & 1 &  1 & -1 & -1 & -1 & -1 &  1 &  1 \cr
              & 1 &  1 & -1 &  1 & -1 &  1 & -1 & -1 \cr
              & 1 &  1 &  1 & -1 &  1 & -1 & -1 & -1 \cr
              & 1 &  1 &  1 &  1 &  1 &  1 &  1 &  1
              },
\end{equation*}
and the treatment combinations are 
$
\bm z_1=(-1, -1, -1),
$
$
\bm z_2=(-1, -1, 1),
$
$
\bm z_3=(-1, 1, -1),
$
$
\bm z_4=(-1, 1, 1),
$
$
\bm z_5=(1, -1, -1),
$
$
\bm z_6=(1, -1, 1),
$
$
\bm z_7=(1, 1, -1),
$
and
$
\bm z_8=(1, 1, 1),
$
respectively.

\subsection{Potential outcomes and factorial effects}

Utilizing the potential outcomes framework \citep{Neyman:1923, Rubin:1974}, \cite{Dasgupta:2015} advocated conducting randomization-based causal inference for $2^K$ factorial designs with $N \ge 2^{K+1}$ units, and invoke the Stable Unit Treatment Value Assumption \citep[SUTVA,][]{Rubin:1980} that there is only one version of the treatment and no interference among the units, for $i = 1, \ldots, N$ we denote the potential outcome of unit $i$ under treatment combination $\bm z_j$ as $Y_i(\bm z_j),$ and
$
\bm Y_i = \{ Y_i(\bm z_1), \ldots, Y_i(\bm z_J) \}^\prime.
$
For binary outcomes $Y_i(\bm z_j) \in \{0, 1\}$ $(i = 1, \ldots, N; j = 1, \ldots, J):$ 
\begin{enumerate}

\item Let
\begin{equation*}
D_{k_1, \ldots, k_J} = \sum_{i=1}^N \prod_{j=1}^J 1_{\{Y_i(\bm z_j) = k_j\}}
\quad
(k_1, \ldots, k_J \in \{0, 1\}).
\end{equation*}
By definition
$
\sum_{k_1=0}^1 \ldots \sum_{k_J=0}^1 D_{k_1, \ldots, k_J} = N.
$
We characterize the potential outcomes using
$
\bm D = (D_{0, 0, \ldots, 0}, D_{0, 0, \ldots, 1}, \ldots, D_{1, 1, \ldots, 0}, D_{1, 1, 
\ldots, 1})^\prime,
$
where the indices are ordered binary representations of 
$
\{0, \ldots, J-1\};
$

\item For all
$
\{j_1, \ldots, j_s\} \subset \{1, \ldots, J\},
$
let
\begin{equation*}
N_{j_1, \ldots, j_s} 
= \sum_{i=1}^N 1_{\left\{Y_i(\bm z_{j_1}) = 1, \ldots, Y_i(\bm z_{j_s}) = 1\right\}}
= \sum_{i=1}^N \prod_{r = 1}^s Y_i(\bm z_{j_r}).
\end{equation*}

\end{enumerate}
Using the new notations, let the average potential outcome for $\bm z_j$ is
$
p_j = N_j / N
$
for $j=1, \ldots, J,$ and let
$
\bm p = ( p_1, \ldots, p_J )^\prime.
$
Therefore, for all units $i = 1, \ldots, N$ and all $l=1, \ldots, J-1,$ we define the $l$th individual-level factorial effect for unit $i$ as 
$
\tau_{il} = 2^{-(K-1)} \bm h_l^\prime \bm Y_i.
$
Consequently, we let the population-level factorial effects be
$
\bar \tau_l = 2^{-(K-1)} \bm h_l^\prime \bm p.
$

\subsection{Randomization-based inference}

Let $n_1, \ldots, n_J$ be positive constants such that
$
\sum n_j = N.
$
For all $j = 1, \ldots, J,$ we randomly assign $n_j \ge 2$ units to $\bm z_j.$ For all $i = 1, \ldots, N$ and all $j = 1, \ldots, J,$ let
$
W_i(\bm z_j) = 1
$
if unit $i$ is assigned to $\bm z_j,$ and zero otherwise, and let
$
\bm W = \{ W_i (\bm z_j) \}_{ij}
$
denote the treatment assignment. Therefore, the observed and missing potential outcomes for unit $i$ are
$
Y_i^\textrm{obs} = \sum_{j=1}^J W_i(\bm z_j) Y_i(\bm z_j)
$
and
$
\bm Y_i^\mathrm{mis} = \{ Y_i(\bm z_j): W_i(\bm z_j) = 0 \},
$
respectively. We denote the observed and missing outcomes for the design as
$
\bm Y^\mathrm{obs} = (Y_1^\textrm{obs}, \ldots, Y_N^\textrm{obs})^\prime
$
and
$
\bm Y^\mathrm{mis} = (\bm Y_1^\mathrm{mis}, \ldots, \bm Y_N^\mathrm{mis})
$
respectively, and
\begin{equation*}
n_j^\textrm{obs} 
= \sum_{i=1}^N W_i(\bm z_j)Y_i(\bm z_j)
= \sum_{i: W_i(\bm z_j) = 1} Y_i^\textrm{obs}
\quad
(j = 1, \ldots, J).
\end{equation*}
The average observed potential outcome for $\bm z_j$ is 
$
\hat p_j = n_j^\mathrm{obs} / n_j,
$
and denote
$
\hat{\bm p}
= 
(
\hat p_1, 
\ldots,
\hat p_J
)^\prime.
$
Consequently, the randomization-based estimators for the factorial effects are
\begin{equation}
\label{eq:factorial-effects-estimator}
\hat {\bar \tau}_l =  2^{-(K-1)} \bm h_l^\prime \hat{\bm p}
\quad
(l = 1, \ldots, J-1).
\end{equation}

Motivated by the seminal work of \cite{Dasgupta:2015}, \cite{Lu:2016a, Lu:2016b} derived the sampling variance of the estimator in \eqref{eq:factorial-effects-estimator} as
\begin{equation}
\label{eq:factorial-effects-variance}
\mathrm{Var}(\hat {\bar \tau}_l) = \frac{1}{2^{2(K-1)}} \sum_{j=1}^J S_j^2 / n_j - \frac{1}{N} S^2(\bar \tau_l),
\end{equation}
where 
\begin{equation*}
S_j^2 
= 
(N-1)^{-1}
\sum_{i=1}^N 
\left( 
Y_i(\bm z_j) - \bar Y(\bm z_1) 
\right)^2 
=
\frac{N}{N-1} p_j (1 - p_j)
\end{equation*}
is the variance of potential outcomes for $\bm z_j$, and 
$
S^2(\bar \tau_l) 
= 
(N-1)^{-1}
\sum_{i=1}^N 
(\tau_{il} - \bar \tau_l)^2
$
is the variance of the $l$th individual-level factorial effects. The ``Neymanian'' estimator for the sampling variance \eqref{eq:factorial-effects-variance} is obtained by substituting $S_j^2$ with its unbiased estimate
\begin{equation*}
s_j^2 
= 
(n_j - 1)^{-1} 
\sum_{i=1}^N
W_i(\bm z_j)
\{
Y_i^{\textrm{obs}} - \bar Y^{\textrm{obs}}(\bm z_j) 
\}^2
=
\frac{n_j}{n_j - 1} \hat p_j (1 - \hat p_j),
\end{equation*}
and substituting 
$
S^2(\bar \tau_l) 
$
with its lower bound 0:
\begin{equation}
\label{eq:factorial-effects-variance-estimator}
\widehat{\mathrm{Var}}_{\mathrm{Ney}}(\hat {\bar \tau}_l) 
= 2^{-2(K-1)} \sum_{j=1}^J s_j^2 / n_j
= 2^{-2(K-1)} \sum_{j=1}^J \frac{\hat p_j (1 - \hat p_j)}{n_j - 1}
\end{equation}
because
$
S^2(\bar \tau_l) 
$ 
is not identifiable from the observed data. This estimator is ``conservative'' in the sense that it over-estimates the true sampling variance on average by
$
\mathrm{E} 
\left\{ 
\widehat{\mathrm{Var}}_{\mathrm{Ney}}(\hat {\bar \tau}_l) 
\right\} 
- \mathrm{Var}(\hat{\bar \tau}_l) 
=  S^2(\bar \tau_l) / N.
$
The bias is generally positive, unless strict additivity \citep{Dasgupta:2015, Ding:2016, Ding:2017} holds, i.e., 
$
\tau_{il} = \tau_{i^\prime l}
$
for all
$
i \ne i^\prime.
$

\section{Finite-population Bayesian analysis}
\label{sec:bayes}

\subsection{Background}

Motivated by the potential deficiencies of Neymanian inference, in this section we extend \cite{Rubin:1978}'s finite-population Bayesian causal inference framework, which is employed by several researchers for treatment-control studies \citep[e.g.,][]{Hirano:2000, Schwartz:2011, Mattei:2013}, to $2^K$ factorial designs.

To ensure that the paper is self-contained, we briefly summarize \cite{Rubin:1978}'s general framework \citep[c.f.][]{Imbens:2015} as follows (we use $f(\cdot)$ and $f(\cdot|\cdot)$ as generic symbols for unconditional and conditional distributions, respectively):
\begin{enumerate}
\item Jointly model the (observed and missing) potential outcomes and treatment assignment by $f(\bm Y^\mathrm{obs}, \bm Y^\mathrm{mis}, \bm W \mid \bm \Theta),$ and specify the prior distribution for the parameters $f(\bm \Theta);$

\item Obtain the posterior distribution of the missing potential outcomes $\bm Y^\mathrm{mis},$ conditioning on the observed data $\bm Y^\mathrm{obs},$ the treatment assignment $\bm W,$ and the parameters $\bm \Theta:$
\begin{equation}
\label{eq:missing-po-imputation-given-theta}
f(\bm Y^\mathrm{mis} \mid \bm Y^\mathrm{obs}, \bm W, \bm \Theta) = \frac{f(\bm Y^\mathrm{obs}, \bm Y^\mathrm{mis} \mid \bm W, \bm \Theta)}{\int_{\bm y^\mathrm{mis}} f(\bm Y^\mathrm{obs}, \bm y^\mathrm{mis} \mid \bm W, \bm \Theta) d\bm y^\mathrm{mis}}
;
\end{equation}

\item Obtain the posterior distribution of the parameters $\bm \Theta,$ conditioning on the missing potential outcomes $\bm Y^\mathrm{mis}$ and the treatment assignment $\bm W:$
\begin{equation}
\label{eq:posterior-theta}
f(\bm \Theta \mid \bm Y^\mathrm{obs}, \bm W) = \frac{f(\bm \Theta) \int_{\bm y^\mathrm{mis}} f(\bm Y^\mathrm{obs}, \bm y^\mathrm{mis}, \bm W \mid \bm \Theta) d\bm y^\mathrm{mis} }{\int_{\bm \theta} \int_{\bm y^\mathrm{mis}} f(\bm \theta)  f(\bm Y^\mathrm{obs}, \bm y^\mathrm{mis}, \bm W \mid \bm \theta) d\bm y^\mathrm{mis} d\bm \theta}
;
\end{equation}

\item Obtain the posterior predictive distribution of $\bm Y^\mathrm{mis}:$
\begin{equation}
\label{eq:missing-po-imputation}
f(\bm Y^\mathrm{mis} \mid \bm Y^\mathrm{obs}, \bm W) = \int_{\bm \theta} f(\bm Y^\mathrm{mis} \mid \bm Y^\mathrm{obs}, \bm W, \bm \theta) f(\bm \theta \mid \bm Y^\mathrm{obs}, \bm W) d\bm \theta
,
\end{equation}
and the posterior predictive distribution of $\bar \tau_l,$ which is a function of $\bm Y^\mathrm{obs}$ and $\bm Y^\mathrm{mis}.$
\end{enumerate}
Under the context of randomized controlled $2^K$ factorial designs, the treatment assignment $\bm W$ is ignorable \citep{Rubin:1978}, implying that we can simplify \eqref{eq:missing-po-imputation-given-theta}--\eqref{eq:missing-po-imputation} by essentially dropping it from the right hand sides. Moreover, SUTVA implies further simplifications of \eqref{eq:missing-po-imputation-given-theta}--\eqref{eq:missing-po-imputation}, as we will show in the next section.


\subsection{Independent potential outcomes model}

Following \cite{Ding:2016}, we first consider a model with independent potential outcomes. For all $j = 1, \ldots, J,$ let 
$
\pi_j = \mathrm{Pr} \{ Y_i(\bm z_j) = 1 \}
$
denote the (prior) marginal probabilities of the potential outcomes. Suppose that the marginal probabilities are independently generated by 
$
\mathrm{Beta}(\alpha_j, \beta_j),
$
where $\alpha_j$ and $\beta_j$ are pre-specified constants. Given
$
\bm \pi_\mathrm{mar} = (\pi_1, \ldots, \pi_J)^\prime
$
assume that the potential outcomes for unit $i = 1, \ldots, N$ are generated by
\begin{equation}
\label{eq:independent-po-model}
Y_i(\bm z_j) \sim \mathrm{Bernoulli}(\pi_j)
\quad
(j = 1, \ldots, J);
\quad
Y_i(\bm z_{j^\prime}) \perp Y_i(\bm z_{j^{\prime\prime}})
\quad
(j^\prime \ne j^{\prime\prime}).
\end{equation}
As mentioned previously, SUTVA and the completely randomized treatment assignment
$
\bm W
$
enable us to derive \eqref{eq:posterior-theta} as follows:
\begin{equation*}
f(\bm \pi_\mathrm{mar} \mid \bm Y^\mathrm{obs}, \bm W) \propto \prod_{j=1}^J (\pi_j)^{\alpha_j - 1} (1 - \pi_j)^{\beta_j - 1} \prod_{j=1}^J (\pi_j)^{n_j^\mathrm{obs}} (1 - \pi_j)^{n_j - n_j^\mathrm{obs}},
\end{equation*}
which immediately suggests the following two-step Monte Carlo procedure to sample from the posterior predictive distribution of the factorial effect $\bar \tau_l:$
\begin{enumerate}
\item Draw $\bm \pi_\mathrm{mar}$ from
\begin{equation}
\label{eq:pi-posterior}
\pi_j \mid \bm Y^\mathrm{obs}, \bm W \stackrel{ind.}{\sim} \mathrm{Beta}
(
\alpha_j + n_j^\mathrm{obs}, \beta_j + n_j - n_j^\mathrm{obs}
)
\quad
(j = 1, \ldots, J);
\end{equation}

\item For all $j=1, \ldots, J,$ let $B_j$ denote the sum of missing potential outcomes for $\bm z_j.$ Given the drawn $\bm \pi_\mathrm{mar},$ draw
\begin{equation}
\label{eq:b-posterior}
B_j 
\mid 
\bm Y^\mathrm{obs}, \bm W, \bm \pi_\mathrm{mar}
\stackrel{ind.}{\sim}
\mathrm{Binomial} (N - n_j, \pi_j),
\end{equation}
and therefore
\begin{equation}
\label{eq:taul-posterior}
\bar \tau_l \mid \bm Y^\mathrm{obs}, \bm W, \bm \pi_\mathrm{mar} \sim  2^{-(k-1)} N^{-1} \sum_{j=1}^J h_{lj} (n_j^\mathrm{obs} + B_j).
\end{equation}
\end{enumerate}

There is a two-fold reason that we consider the independent potential outcomes model as the first step of applying \cite{Rubin:1978}'s finite-population Bayesian causal inference framework to $2^K$ factorial designs. On the one hand, because of the missing data nature of the potential outcomes framework \citep{Imbens:2015}, the observed data only directly helps us infer the marginal distributions of but not the associations between the potential outcomes. On the other hand, the imputation procedure \eqref{eq:pi-posterior}--\eqref{eq:taul-posterior} implies closed-form expressions for the posterior predictive mean and variance of $\bar \tau_l.$

\begin{theorem}
\label{thm:bayes-indep}
Let
$
n_j^\prime = n_j + \alpha_j + \beta_j
$
and
$
\hat p_j^\prime = (n_j^\mathrm{obs} + \alpha_j) / n_j^\prime
$
for all $j = 1, \ldots, J.$ The posterior predictive mean and variance of $\bar \tau_l$ are
\begin{equation}
\label{eq:taul-posterior-mean}
E(\bar \tau_l \mid \bm Y^\mathrm{obs}, \bm W)
=
2^{-(k-1)} N^{-1} 
\sum_{j=1}^J h_{lj} 
\{
n_j \hat p_j
+
(N - n_j) \hat p_j^\prime
\}
\end{equation}
and
\begin{equation}
\label{eq:factorial-effects-variance-estimator-bayes}
\mathrm{Var}(\bar \tau_l \mid \bm Y^\mathrm{obs}, \bm W)
=
2^{-2(k-1)}
\sum_{j=1}^J 
\frac{N - n_j + n_j^\prime}{N}
\left(
1 - \frac{n_j}{N}
\right)
\frac{\hat p_j^\prime (1 - \hat p_j^\prime)}{n_j^\prime + 1},
\end{equation}
respectively. 
\end{theorem}

\begin{corollary}
\label{coro:bayes-indep-approx}
Assume that
$
\alpha_j, \beta_j \ll n_j
$
for all
$
j=1, \ldots, J,
$
then
\begin{equation}
\label{eq:taul-posterior-mean-approx}
E(\bar \tau_l \mid \bm Y^\mathrm{obs}, \bm W)
\approx \hat{\bar \tau}_l
\end{equation}
and
\begin{equation}
\label{eq:factorial-effects-variance-estimator-bayes-approx}
\mathrm{Var}(\bar \tau_l \mid \bm Y^\mathrm{obs}, \bm W)
\approx
2^{-2(K-1)} \sum_{j=1}^J \left( 1 - \frac{n_j}{N} \right) \frac{\hat p_j (1 - \hat p_j)}{n_j - 1},
\end{equation}
respectively. The approximations become exact as $N \rightarrow \infty.$ 
\end{corollary}

We conclude this section by following \cite{Rubin:1984} and evaluating the Frequentist property of the above Bayesian procedure. Among other things, the following corollary suggests that when the potential outcomes are independent, the posterior predictive variance of $\bar \tau_l$ in \eqref{eq:factorial-effects-variance-estimator-bayes} reconciles with its Frequentist counterpart. 

\begin{corollary}
\label{coro:frequentist-bayesian-recon}
The posterior predictive variance of the factorial effect $\bar \tau_l$ in \eqref{eq:factorial-effects-variance-estimator-bayes} is generally smaller than the Neymanian variance estimator in \eqref{eq:factorial-effects-variance-estimator}, i.e.,
\begin{equation*}
\mathrm{Var}(\bar \tau_l \mid \bm Y^\mathrm{obs}, \bm W)
\le 
\widehat{\mathrm{Var}}_{\mathrm{Ney}}(\hat {\bar \tau}_l). 
\end{equation*}
The equality holds if all potential outcomes are pair-wisely unassociated:
\begin{equation}
\label{eq:po-no-association}
S_{jj^\prime} = (N - 1)^{-1} \sum_{i=1}^N 
\{ 
Y_i(\bm z_j) - p_j
\}
\{
Y_i(\bm z_{j^\prime}) - p_{j^\prime}
\} 
= 0
\quad
(j \ne j^\prime).
\end{equation}
\end{corollary}

\subsection{Sensitivity analysis}

Despite the apparent theoretical and computational appeals, the aforementioned independent potential outcomes model may be inappropriate in practice, as pointed out by \cite{Ding:2016}. In particular, when the potential outcomes are positively correlated, the resulted Bayesian credible intervals may under-cover the factorial effect $\bar \tau_l.$ Therefore, even though the marginal distributions of the potential outcomes can be inferred, it is imperative that we take into account the dependence structure between them, when developing any Bayesian procedures for $2^K$ factorial designs. To facilitate more in-depth understanding, we discuss the key role that the independence assumption in \eqref{eq:independent-po-model} plays, before presenting any proposals.

There are two pain-points we wish to emphasize here. First, with or without the independence assumption, the posterior distribution of the marginal probabilities 
$
\bm \pi_\mathrm{mar} = (\pi_1, \ldots, \pi_J)^\prime
$
and the posterior predictive mean of $\bar \tau_l$ remain the same as in \eqref{eq:pi-posterior} and \eqref{eq:taul-posterior-mean}, respectively. Second and more importantly, as mentioned before the crux of \cite{Rubin:1978}'s framework is the imputation of the missing potential outcomes. To be specific, for each $i = 1, \ldots, N,$ because there exists only one $j$ such that 
$
W_i(\bm z_j) = 1
$
and
$
Y_i^\mathrm{obs} = Y_i(\bm z_j),
$
we need to impute
$
Y_i(\bm z_{j^\prime})
$
for all $j^\prime \ne j.$ This is rather straightforward under the independence assumption, because as mentioned in the previous section we can draw the marginal probabilities 
$
\bm \pi_\mathrm{mar}
$
from \eqref{eq:pi-posterior}, and then draw
$
Y_i(\bm z_{j^\prime}) \sim \mathrm{Bernoulli} (\pi_{j^\prime})  
$
for all
$
j^\prime \ne j.
$
Unfortunately, however, this strategy no longer works without the independence assumption, because the value of the observed potential outcome 
$
Y_i(\bm z_j)
$
indeed becomes relevant when imputing the missing potential outcomes, as pointed out by \cite{Ding:2016}. To be more specific, denote the conditional probabilities
\begin{equation}
\label{eq:po-conditional-distribution}
\pi_{j^\prime \mid j = s} = \mathrm{Pr} \{ Y_i(\bm z_{j^\prime}) = 1 \mid Y_i(\bm z_j) = s \}
\end{equation}
for $s = 0, 1.$ If
$
Y_i(\bm z_j) = s,
$
the missing potential outcome
$
Y_i(\bm z_{j^\prime})
\sim 
\mathrm{Bernoulli} (\pi_{j^\prime \mid j = s}).
$

Although the conditional probabilities in \eqref{eq:po-conditional-distribution} are crucial in imputing the missing potential outcomes, they are not identifiable from the observed data, because we cannot jointly observe the potential outcomes under $\bm z_j$ and $\bm z_{j^\prime}.$ For treatment-control studies, \cite{Ding:2016} pointed out that the joint distribution of the treatment and control potential outcomes can be uniquely determined by their marginal distributions and a single association parameter, and proposed to conduct sensitivity analysis by varying the association parameter accordingly. For more general $2^K$ factorial designs, in principle it is possible to fix the marginal probabilities 
$
\bm \pi_\mathrm{mar},
$
and vary the associations between all the potential outcome pairs. However, because the dependence structure becomes more complex \citep{Cox:1972, Teugels:1990, Dai:2013}, it is imperative to invoke some structural assumptions to make this problem somewhat tractable. From the lengthy list of proposals \citep[e.g.,][]{Emrich:1991, Lee:1993, Gange:1995, Park:1996, Kang:2001, Oman:2001, Qaqish:2003}, we adopt the framework by \cite{Oman:2009}, who proposed to construct the joint distribution of the potential outcomes such that
\begin{equation*}
\mathrm{Pr} \{Y_i(\bm z_j) = 1, Y_i(\bm z_{j^\prime}) = 1\}
=
(1 - \gamma_{jj^\prime}) \pi_j \pi_{j^\prime} + \gamma_{jj^\prime} \min(\pi_j, \pi_{j^\prime}),
\end{equation*}
for all $j \ne j^\prime,$ where $\gamma_{jj^\prime} \in [0, 1)$ characterizes the association between the potential outcome pair $Y_i(\bm z_j)$ and $Y_i(\bm z_{j^\prime}).$ The above suggests that, for any fixed value of $\gamma_{jj^\prime},$ we can derive closed-form expressions for the conditional probabilities in \eqref{eq:po-conditional-distribution}. 

\begin{theorem}
\label{thm:conditional-prob}
Under \cite{Oman:2009}'s framework,
\begin{equation}
\label{eq:po-conditional-distribution-oman}
\pi_{j^\prime \mid j = 1} 
= 
(1 - \gamma_{jj^\prime}) \pi_{j^\prime} + \gamma_{jj^\prime} \min \left( 1, \frac{\pi_{j^\prime}}{\pi_j} \right),
\quad
\pi_{j^\prime \mid j = 0} 
=
(1 - \gamma_{jj^\prime}) \pi_{j^\prime} + \gamma_{jj^\prime} \frac{\max(\pi_{j^\prime} - \pi_j, 0)}{1 - \pi_j}.
\end{equation}
\end{theorem}

Theorem \ref{thm:conditional-prob} suggests that, in order to perform the sensitivity analysis, we only need to specify $\gamma_{jj^\prime}$ for all $j \ne j^\prime,$ i.e., the pair-wise correlation structure of the potential outcomes. \cite{Oman:2009} proposed several models for such correlation structure. For $2^2$ factorial designs which we focus in the next two sections, we adopt the AR(1) correlation structure, where we specify the sensitivity parameter $\rho \in [0, 1),$ and for all 
$
j \ne j^\prime
$
let
$
\gamma_{jj^\prime} = \rho^{|j - j^\prime|}.
$
This appears to be a reasonable assumption for the dependence structure of the potential outcomes. To be more specific, because the treatment combinations $\bm z_1, \ldots, \bm z_J$ are nonexchangeable by definition, and we are essentially assuming that the association between $Y_i(\bm z_j)$ and $Y_i(\bm z_{j^\prime})$ exponentially decays as $|j - j^\prime|$ (i,e., the ``distance'' between $\bm z_j$ and $\bm z_{j^\prime}$) increases. However, for more general (i.e., $K \ge 3$) factorial designs, we may need to consider other dependence structures \citep[e.g., Toeplitz matrix, see][]{Chen:2006}.

With the help of Theorem \ref{thm:conditional-prob} and the pair-wise correlation structure of the potential outcomes, we now formally present the Bayesian sensitivity analysis procedure as follows: 
\begin{enumerate}

\item Specify the value of the sensitivity parameter $\rho;$ 

\item Same as for the independent potential outcomes model, use \eqref{eq:pi-posterior} to draw the marginal probabilities
$
\bm \pi_\mathrm{mar} = (\pi_1, \ldots, \pi_J)^\prime;
$ 

\item For all $j = 1, \ldots, J,$ use \eqref{eq:po-conditional-distribution-oman} to calculate the conditional probabilities in \eqref{eq:po-conditional-distribution}; 

\item For all $j\prime \ne j$ independently draw
\begin{equation*}
B_{j\mid j^\prime = 1} \sim \mathrm{Binomial} (n_{j^\prime}^\mathrm{obs}, \pi_{j \mid j^\prime = 1}),
\quad
B_{j\mid j^\prime = 0} \sim \mathrm{Binomial} (n_{j^\prime} - n_{j^\prime}^\mathrm{obs}, \pi_{j \mid j^\prime = 0})
\end{equation*}
and let
\begin{equation*}
C_j = \sum_{j^\prime \ne j} \sum_{s = 0}^1 B_{j \mid j^\prime = s} 
\end{equation*}
denote the sum of missing potential outcomes for $\bm z_j.$ Therefore
\begin{equation*}
\bar \tau_l \mid \bm Y^\mathrm{obs}, \bm W, \bm \pi_\mathrm{mar} \sim  2^{-(k-1)} N^{-1} \sum_{j=1}^J h_{lj}(n_j^\mathrm{obs} + C_j).
\end{equation*}

\end{enumerate}

For fixed value of the sensitivity parameter $\gamma,$ when closed-form expressions for the posterior predictive mean and variance of $\bar \tau_l$ are infeasible, we use Monte Carlo methods for approximation. As suggested by \cite{Ding:2016}, in practice we can vary $\rho$ over a wide range of values (e.g., from zero to one), and repeat the above Monte Carlo procedure for each value. 
In the next two sections, we provide several simulated and empirical examples for illustration.


\section{Simulation studies}
\label{sec:simu}

In this section, we conduct simulation studies to examine the Neymanian variance estimator in \eqref{eq:factorial-effects-variance-estimator} and the posterior predictive variance under independence assumption in \eqref{eq:factorial-effects-variance-estimator-bayes}.

To mimic the empirical examples in the next section, consider a balanced $2^2$ factorial design with $N = 800$ experimental units, so that 
$
(n_1, n_2, n_3, n_4) = (200, 200, 200, 200).
$
To save space for the main text, we focus on
$
\bar \tau_1.
$
We generate 100 simulation cases by repeatedly drawing from the following hierarchical model: 
\begin{equation*}
U_j \stackrel{\mathrm{iid.}}{\sim} \mathrm{Unif} (0, 1)
\quad
(j = 1, \ldots, 16);
\quad
\bm \tau = (U_1, \ldots, U_{16})^\prime \big/ \sum_{j=1}^{16} U_j,     
\end{equation*}
and
\begin{equation*}
\bm D = (D_{0, 0, \ldots, 0}, \ldots, D_{1, 1, 
\ldots, 1})^\prime \mid \bm \tau \sim \mathrm{Multinomial} (800, \bm \tau).
\end{equation*} 
We report details of the simulation cases in Appendix \ref{subsec:simu-cases}, so that the readers can replicate our simulation results. For each case, we first calculate the factorial effect
$
\bar \tau_1.
$
Second, independently draw $500$ treatment assignments and the corresponding observed data. Third, For each observed dataset, use \eqref{eq:factorial-effects-estimator}, \eqref{eq:factorial-effects-variance-estimator} and \eqref{eq:factorial-effects-variance-estimator-bayes} to calculate the point estimate $\hat{\bar \tau}_1,$ its Neymanian variance estimates, and the posterior predictive variance of $\bar \tau_1$ under the independence assumption, respectively. Fourth, construct the 95\% Neymanian confidence intervals and ``independent'' Bayesian credible interval.

Figure \ref{fg:simu1} contains the coverage rates of the intervals. The Neymanian interval generally over-covers $\bar \tau_1,$ with coverage rates greater than 0.96 for 100\% of the cases. Second, the independent Bayesian interval manages to correct the over-coverage of the Neymanian interval, with coverage rates greater than 0.96 for only 9\% of the cases. However, for 11\% of the cases it under-covers with coverage rates smaller than 0.94, suggesting that the proposed Bayesian sensitivity analysis is indeed necessary. 

\begin{figure}[htbp]
\centering
\includegraphics[width = .8\linewidth]{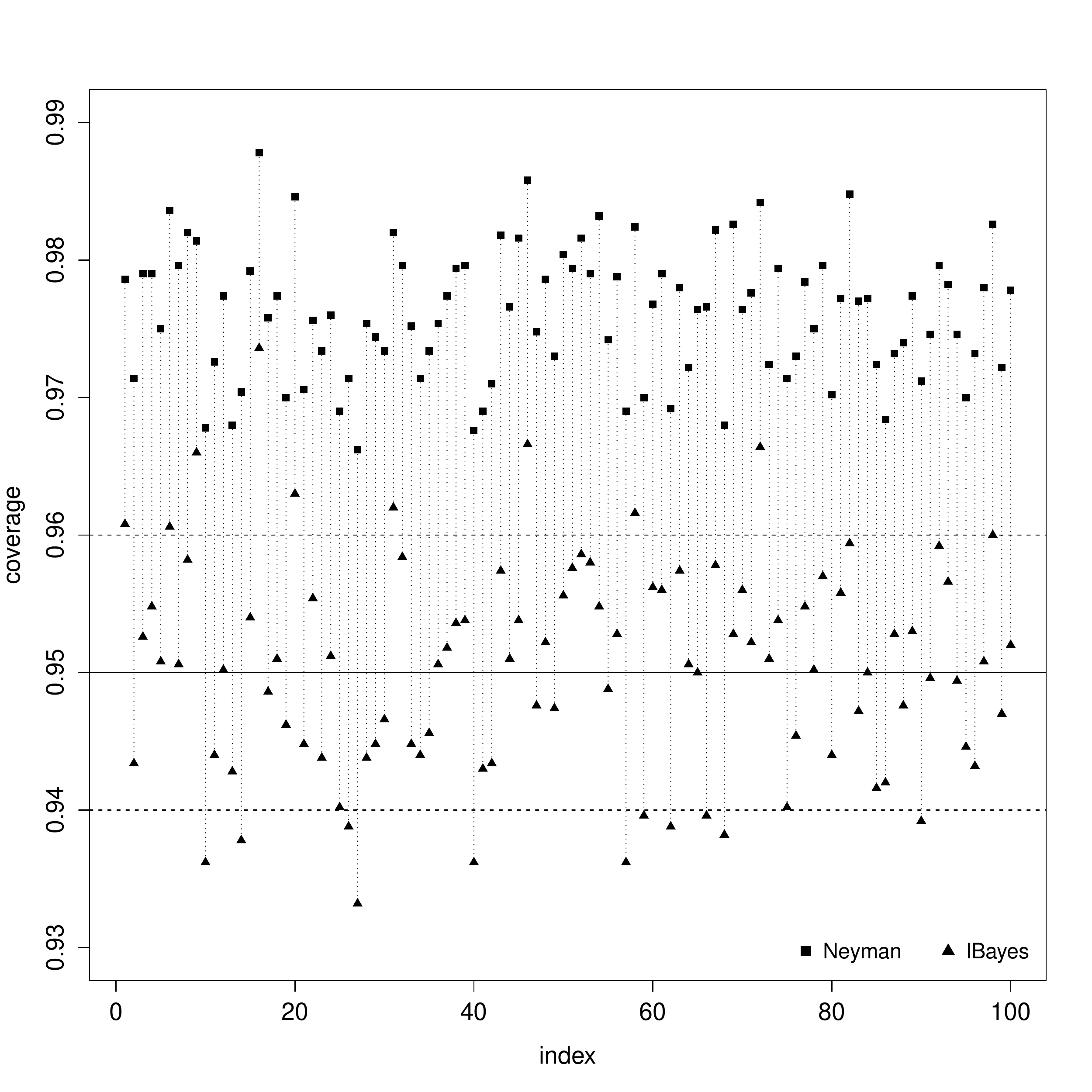}
\caption{Simulation results: The horizontal axis represents the case index, and the vertical shows the coverage rates for the 95\% Neymanian (rectangular) and independent Bayesian (triangular) intervals.}
\label{fg:simu1}
\end{figure}

To more thoroughly demonstrate the characteristic of the Bayesian interval, in Appendix \ref{subsec:add-simu} we conduct a new set of simulation studies.

\section{Empirical example}
\label{sec:example}

We re-analyze a randomized controlled $2^2$ factorial design which evaluated the factorial effects of nicotine gum consumption (2gm/day vs. placebo) and counseling (health education vs. motivational interviewing) on $N = 755$ African American light smokers. For details of the study, see \cite{Ahluwalia:2006}. The primary outcome of the study is whether participants quit smoking 26 weeks after enrollment, and the observed data is 
$
(n_1, n_2, n_3, n_4) = (189, 188, 189, 189)
$
and
$
(n_1^\mathrm{obs}, n_2^\mathrm{obs}, n_3^\mathrm{obs}, n_4^\mathrm{obs})
= (13, 29, 19, 34).
$

To save space for the main text, we only focus on 
$
\bar \tau_2.
$
We report the results in Figure \ref{fg:empirical}, from which we can draw several conclusions. First, from a Neymanian perspective, $\hat{\bar \tau}_2 = 0.082$ and the corresponding 95\% confidence interval is (0.035, 0.129). Second, the independence Bayesian interval is (0.041, 0.123), which is 14\% narrower than the Neymanian interval. Third, the sensitivity analysis suggests that the widest Bayesian interval is (0.037, 0.125), where $\rho = 0.68.$ In other words, this is our most ``conservative'' Bayesian interval without knowing the underlying correlation between the potential outcomes. 
\begin{figure}[htbp]
\centering
\includegraphics[width = 1\linewidth]{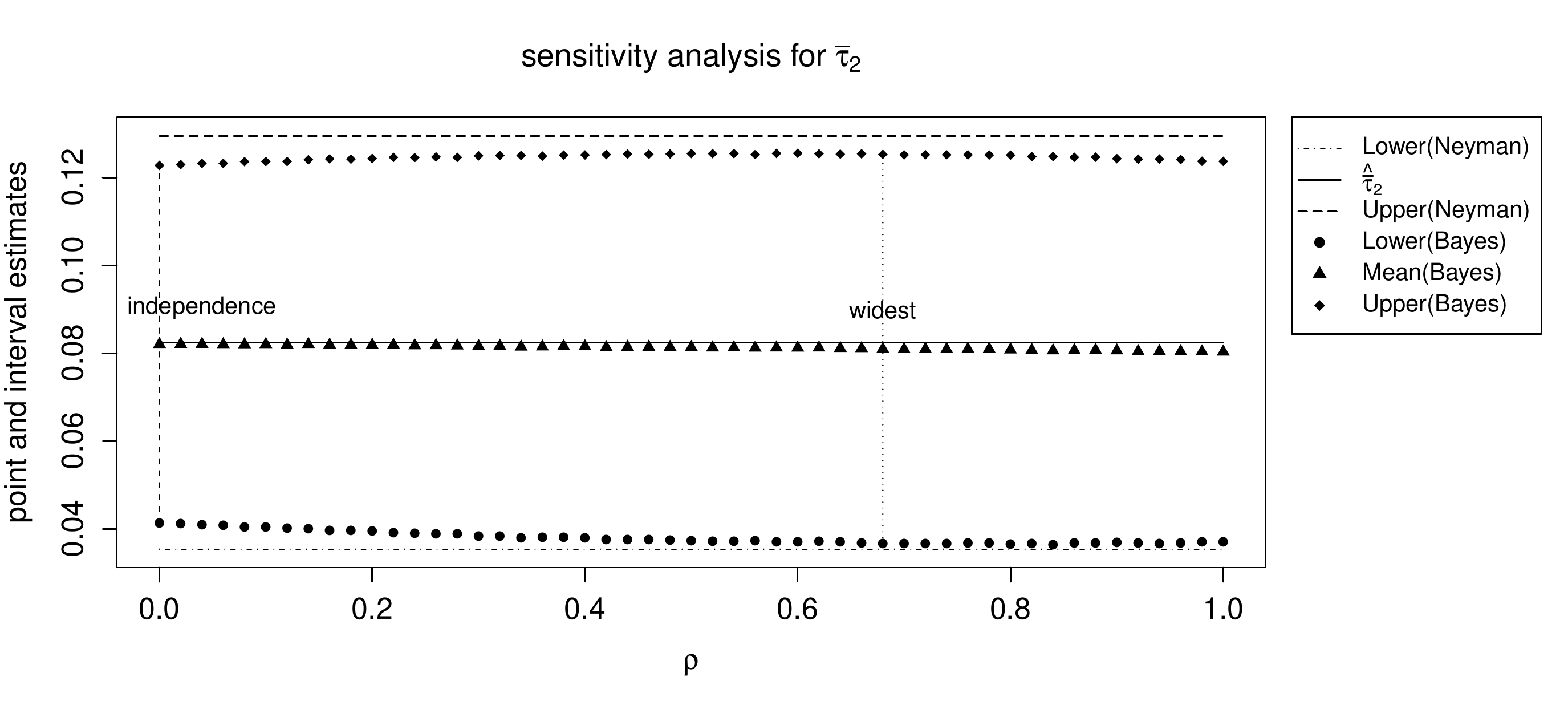}
\caption{Empirical results for \cite{Ahluwalia:2006}'s data-set: the Neymanian, ``independent'' Bayesian and ``conservative'' Bayesian point and interval estimates.}
\label{fg:empirical}
\end{figure}

\section{Concluding remarks}
\label{sec:conclusion}

To address the (sometimes severe) variance over-estimation issue of the classic Neymanian causal inference framework, this paper extended \cite{Rubin:1978}'s and \cite{Ding:2016}'s finite-population Bayesian inference framework to $2^K$ factorial designs with binary outcomes. As emphasized by \cite{Rubin:1978}, the crux of the finite-population Bayesian framework is the imputation of the missing potential outcomes. Because the potential outcomes cannot be jointly observed, we first developed an imputation model under the assumption that they are independent given their marginal probabilities. To assess how violations of the independence assumption impacted our analysis, we proposed a novel sensitivity analysis procedure. To demonstrate the advantages of our proposed methodology, we provided several simulated and empirical examples.

Our work suggests multiple future directions. First, we can generalize our current framework to more complex experiment settings, such as $3^k$ or fractional factorial designs, and cross-over designs. Second, it is possible to extend our framework to accommodate more general outcomes, such as continuous or time to event. Third, while developing Bayesian procedures is important, it might also be desirable to sharpen the existing Neymanian inference for $2^K$ factorial designs. For treatment-control studies, \cite{Ding:2016} and \cite{Aronow:2014} proposed the respective ``improved'' Neymanian variance estimators, by deriving sharp lower bounds for the individual-level treatment effect variation. Unfortunately, however, extending their results to $2^K$ factorial designs might not be a trivial task, because of the the complex dependence structure of the potential outcomes. Fourth, we can incorporate pre-treatment covariate information into our current framework, especially for developing alternative sensitivity analysis procedures. We leave the above for future research.

\section*{Acknowledgement}

The author thanks the Editor, Associate Editor and five anonymous reviewers for their valuable comments, which improve the quality of this paper significantly. The author benefits from early discussions with Professor Tirthankar Dasgupta at Rutgers and Professor Peng Ding at Berkeley.

\bibliographystyle{apalike}
\bibliography{factorial_binary}

\newpage
\appendix

\section{Proofs of Lemmas, Theorems and Corollaries}

\begin{proof}[Proof of Theorem \ref{thm:bayes-indep}]
The proof largely follow that of \cite{Ding:2016}'s Theorem 2. To be specific, by \eqref{eq:pi-posterior}
\begin{equation}
\label{eq:thm-bayes-indep-1}
E(
\pi_j \mid \bm Y^\mathrm{obs}, \bm W 
) = \hat p_j^\prime,
\quad
\mathrm{Var}(
\pi_j \mid \bm Y^\mathrm{obs}, \bm W 
)
=
\frac{\hat p_j^\prime (1 - \hat p_j^\prime)}{n_j^\prime + 1},
\end{equation}
and therefore
\begin{equation}
\label{eq:thm-bayes-indep-2}
E\{ \pi_j (1 - \pi_j ) \mid \bm Y^\mathrm{obs}, \bm W \}
=
\frac{n_j^\prime}{n_j^\prime + 1} \hat p_j^\prime (1 - \hat p_j^\prime).
\end{equation}

With the help of \eqref{eq:thm-bayes-indep-1}--\eqref{eq:thm-bayes-indep-2}, we can now prove Theorem \ref{thm:bayes-indep}. First, by \eqref{eq:b-posterior}, \eqref{eq:taul-posterior} and \eqref{eq:thm-bayes-indep-1}
\begin{align*}
E(\bar \tau_l \mid \bm Y^\mathrm{obs}, \bm W)
& = E\{E(\bar \tau_l \mid \bm Y^\mathrm{obs}, \bm W, \bm \pi_\mathrm{mar}) \mid \bm Y^\mathrm{obs}, \bm W \} \\
& = 2^{-(k-1)} N^{-1} 
\sum_{j=1}^J h_{lj} 
\{
n_j^\mathrm{obs}
+
(N - n_j) E(\pi_j \mid \bm Y^\mathrm{obs}, \bm W)
\} \\
& = 2^{-(k-1)} N^{-1} 
\sum_{j=1}^J h_{lj} 
\{
n_j \hat p_j
+
(N - n_j) \hat p_j^\prime
\}.
\end{align*}
Second, by \eqref{eq:b-posterior}, \eqref{eq:taul-posterior} and \eqref{eq:thm-bayes-indep-2}
\begin{align*}
E\{\mathrm{Var}(\bar \tau_l \mid \bm Y^\mathrm{obs}, \bm W, \bm \pi_\mathrm{mar}) \mid \bm Y^\mathrm{obs}, \bm W \} 
& = 
2^{-2(k-1)} N^{-2}
\sum_{j=1}^J (N - n_j) E\{ \pi_j (1 - \pi_j)
\mid \bm Y^\mathrm{obs}, \bm W \} \\
& = 
2^{-2(k-1)} N^{-2}
\sum_{j=1}^J \frac{(N - n_j) n_j^\prime}{n_j^\prime + 1} \hat p_j^\prime (1 - \hat p_j^\prime),
\end{align*}
and
\begin{align*}
\mathrm{Var} \{E(\bar \tau_l \mid \bm Y^\mathrm{obs}, \bm W, \bm \pi_\mathrm{mar}) \mid \bm Y^\mathrm{obs}, \bm W \} 
& = 
2^{-2(k-1)} N^{-2}
\sum_{j=1}^J \mathrm{Var}
\left\{ 
\sum_{j=1}^J h_{lj} (N - n_j) \pi_j
\mid \bm Y^\mathrm{obs}, \bm W 
\right\} \\
& = 2^{-2(k-1)} N^{-2}
\sum_{j=1}^J \frac{(N - n_j)^2 \hat p_j^\prime (1 - \hat p_j^\prime)}{n_j^\prime + 1}.
\end{align*}
Consequently,
\begin{align*}
\mathrm{Var}(\bar \tau_l \mid \bm Y^\mathrm{obs}, \bm W)
& = 
E\{\mathrm{Var}(\bar \tau_l \mid \bm Y^\mathrm{obs}, \bm W, \bm \pi_\mathrm{mar}) \mid \bm Y^\mathrm{obs}, \bm W \} 
+
\mathrm{Var} \{E(\bar \tau_l \mid \bm Y^\mathrm{obs}, \bm W, \bm \pi_\mathrm{mar}) \mid \bm Y^\mathrm{obs}, \bm W \} 
\\
& = 2^{-2(k-1)}
\sum_{j=1}^J 
\frac{N - n_j + n_j^\prime}{N}
\left(
1 - \frac{n_j}{N}
\right)
\frac{\hat p_j^\prime (1 - \hat p_j^\prime)}{n_j^\prime + 1}.
\end{align*}
The proof is complete.
\end{proof}

\bigskip
\begin{proof}[Proof of Corollary \ref{coro:bayes-indep-approx}]
Because
$
\alpha_j, \beta_j \ll n_j
$
for $j=1, \ldots, J,$
$$
\hat p_j^\prime 
=
\frac{n_j^\mathrm{obs} + \alpha_j}{n_j + \alpha_j + \beta_j}
\approx 
\hat p_j
$$
and
$$
\frac{N - n_j + n_j^\prime}{N} 
= 
1 + \frac{\alpha_j + \beta_j}{N}
\approx 
1.
$$
Therefore, by \eqref{eq:factorial-effects-estimator} and \eqref{eq:taul-posterior-mean},
\begin{align*}
E(\bar \tau_l \mid \bm Y^\mathrm{obs}, \bm W)
& = 
2^{-(k-1)} N^{-1} 
\sum_{j=1}^J h_{lj} 
\{
n_j \hat p_j
+
(N - n_j) \hat p_j^\prime
\} 
\\
& \approx 
2^{-(k-1)}
\sum_{j=1}^J h_{lj} 
\hat p_j
\\
& = \hat{\bar \tau}_l.
\end{align*}
Similarly, by \eqref{eq:factorial-effects-variance-estimator-bayes},
\begin{align*}
\mathrm{Var}(\bar \tau_l \mid \bm Y^\mathrm{obs}, \bm W)
& = 2^{-2(k-1)}
\sum_{j=1}^J 
\frac{N - n_j + n_j^\prime}{N}
\left(
1 - \frac{n_j}{N}
\right)
\frac{\hat p_j^\prime (1 - \hat p_j^\prime)}{n_j^\prime + 1} \\
& \approx 2^{-2(k-1)}
\sum_{j=1}^J 
\left(
1 - \frac{n_j}{N}
\right)
\frac{\hat p_j (1 - \hat p_j)}{n_j - 1}.
\end{align*}
The proof is complete.
\end{proof}

\bigskip
\begin{proof}[Proof of Corollary \ref{coro:frequentist-bayesian-recon}]
The first part is obvious by \eqref{eq:factorial-effects-variance-estimator}, \eqref{eq:factorial-effects-variance-estimator-bayes}, and the fact that 
\begin{equation*}
1 - \frac{n_j}{N} \le 1
\quad
(j = 1, \ldots, J).
\end{equation*}
Moreover, the definition of $\hat{\bar \tau}_l$ in \eqref{eq:factorial-effects-estimator} suggests that
\begin{align*}
\mathrm{Var}_\mathrm{Ney} ( \hat{\bar \tau}_l ) 
& = 2^{-2(K-1)} 
\left\{
\sum_{j=1}^J \mathrm{Var}_\mathrm{Ney}(\hat p_j) + \sum_{j \ne j^\prime} h_{lj}h_{lj^\prime} \mathrm{Cov}_\mathrm{Ney}(\hat p_j, \hat p_{j^\prime})
\right\} \\
& = 2^{-2(K-1)} 
\left\{ 
\sum_{j=1}^J \left(\frac{1}{n_j} - \frac{1}{N} \right) S_j^2 - \frac{1}{N} \sum_{j \ne j^\prime} h_{lj}h_{lj^\prime} S_{jj^\prime}
\right\} \\
& = 2^{-2(K-1)} 
\sum_{j=1}^J \left(\frac{1}{n_j} - \frac{1}{N} \right) S_j^2.
\end{align*}
The last step holds because of \eqref{eq:po-no-association}. Therefore, the corresponding
\begin{align*}
\widehat{\mathrm{Var}}_{\mathrm{Ney}}(\hat {\bar \tau}_l)
& = 2^{-2(K-1)} 
\sum_{j=1}^J \left(\frac{1}{n_j} - \frac{1}{N} \right) s_j^2 \\
& = 2^{-2(K-1)} \sum_{j=1}^J \left( 1 - \frac{n_j}{N} \right) \frac{\hat p_j (1 - \hat p_j)}{n_j - 1},
\end{align*}
which completes the proof.
\end{proof}

\bigskip
\begin{proof}[Proof of Theorem \ref{thm:conditional-prob}]
Because
\begin{equation*}
\mathrm{Pr} \{Y_i(\bm z_j) = 1, Y_i(\bm z_{j^\prime}) = 1\}
=
(1 - \gamma_{jj^\prime}) \pi_j \pi_{j^\prime} + \gamma_{jj^\prime} \min(\pi_j, \pi_{j^\prime}),
\end{equation*}
we have
\begin{align*}
\pi_{j^\prime \mid j = 1} 
& = \frac{\mathrm{Pr} \{ Y_i(\bm z_{j^\prime}) = 1, Y_i(\bm z_j) = 1 \}}{\mathrm{Pr} \{ Y_i(\bm z_j) = 1 \}} \\
& = \frac{(1 - \gamma_{jj^\prime}) \pi_j \pi_{j^\prime} + \gamma_{jj^\prime} \min(\pi_j, \pi_{j^\prime})}{\pi_j} \\
& = (1 - \gamma_{jj^\prime}) \pi_{j^\prime} + \gamma_{jj^\prime} \min \left( 1, \frac{\pi_{j^\prime}}{\pi_j} \right),
\end{align*}
and
\begin{align*}
\pi_{j^\prime \mid j = 0} 
& = \frac{\mathrm{Pr} \{ Y_i(\bm z_{j^\prime}) = 1, Y_i(\bm z_j) = 0 \}}{\mathrm{Pr} \{ Y_i(\bm z_j) = 0 \}} \\
& = \frac{\mathrm{Pr} \{ Y_i(\bm z_{j^\prime}) = 1 \} - \mathrm{Pr} \{ Y_i(\bm z_{j^\prime}) = 1, Y_i(\bm z_j) = 1 \}}{1 - \mathrm{Pr} \{ Y_i(\bm z_j) = 1 \}} \\
& = \frac{\pi_{j^\prime} - (1 - \gamma_{jj^\prime}) \pi_j \pi_{j^\prime} - \gamma_{jj^\prime} \min(\pi_j, \pi_{j^\prime})}{1 - \pi_j}\\
& = \frac{(1 - \gamma_{jj^\prime})(1 - \pi_j)\pi_{j^\prime} + \gamma_{jj^\prime} \pi_{j^\prime} - \gamma_{jj^\prime} \min(\pi_j, \pi_{j^\prime})}{1 - \pi_j} \\
& = (1 - \gamma_{jj^\prime}) \pi_{j^\prime} + \gamma_{jj^\prime} \frac{\max(\pi_{j^\prime} - \pi_j, 0)}{1 - \pi_j}.
\end{align*}
The proof is complete.
\end{proof}

\section{Additional information of the simulated and empirical examples}

\subsection{Details of the simulation cases}
\label{subsec:simu-cases}

We report the 100 simulation cases in the following three tables:

\begin{itemize}

\item Table 1, Cases 1--33:

\begin{table}[htbp]
\centering
\small
\begin{tabular}{cc}
 Case & $(D_{0, 0, \ldots, 0}; D_{0, 0, \ldots, 1}; \ldots, D_{1, 1, \ldots, 0}; D_{1, 1, 
\ldots, 1})$ \\ 
  \hline
  1 & (33, 12, 0, 63, 18, 93, 63, 118, 53, 41, 44, 71, 67, 58, 58, 8) \\ 
  2 & (52, 61, 10, 57, 111, 64, 22, 25, 11, 67, 85, 39, 7, 107, 57, 25) \\ 
  3 & (30, 79, 46, 26, 103, 31, 94, 130, 29, 9, 75, 1, 50, 34, 42, 21) \\ 
  4 & (50, 140, 0, 73, 22, 58, 0, 93, 128, 23, 22, 10, 3, 51, 93, 34) \\ 
  5 & (61, 47, 89, 91, 92, 49, 30, 7, 46, 50, 9, 24, 7, 66, 22, 110) \\ 
  6 & (51, 70, 58, 89, 32, 6, 59, 98, 7, 77, 43, 65, 113, 0, 27, 5) \\ 
  7 & (11, 118, 61, 54, 23, 24, 5, 77, 62, 15, 110, 34, 76, 8, 8, 114) \\ 
  8 & (66, 16, 73, 49, 75, 26, 34, 24, 23, 28, 74, 92, 88, 29, 52, 51) \\ 
  9 & (17, 71, 35, 45, 55, 16, 23, 25, 3, 87, 106, 64, 90, 80, 65, 18) \\ 
  10 & (21, 100, 37, 11, 105, 99, 5, 1, 99, 20, 74, 30, 18, 18, 55, 107) \\ 
  11 & (62, 55, 54, 48, 60, 36, 60, 55, 67, 80, 36, 24, 23, 48, 43, 49) \\ 
  12 & (50, 59, 95, 15, 8, 0, 65, 32, 69, 29, 67, 46, 57, 93, 60, 55) \\ 
  13 & (36, 99, 70, 68, 15, 97, 2, 28, 20, 75, 70, 73, 0, 34, 32, 81) \\ 
  14 & (84, 65, 66, 5, 70, 23, 7, 24, 2, 71, 93, 19, 62, 40, 93, 76) \\ 
  15 & (86, 107, 17, 106, 14, 30, 74, 19, 11, 64, 50, 3, 41, 12, 92, 74) \\ 
  16 & (14, 29, 38, 123, 11, 33, 18, 46, 65, 41, 12, 115, 112, 21, 24, 98) \\ 
  17 & (27, 61, 47, 35, 13, 83, 44, 56, 88, 66, 24, 52, 22, 57, 54, 71) \\ 
  18 & (10, 7, 65, 75, 1, 63, 64, 79, 33, 103, 60, 23, 63, 76, 13, 65) \\ 
  19 & (34, 10, 92, 21, 2, 72, 93, 7, 51, 65, 44, 65, 70, 64, 73, 37) \\ 
  20 & (25, 102, 88, 54, 57, 75, 14, 31, 96, 19, 26, 48, 71, 92, 2, 0) \\ 
  21 & (8, 61, 78, 40, 35, 85, 75, 78, 49, 0, 64, 52, 41, 39, 23, 72) \\ 
  22 & (48, 79, 87, 28, 28, 6, 52, 53, 75, 20, 71, 29, 49, 18, 87, 70) \\ 
  23 & (41, 18, 62, 35, 1, 74, 51, 62, 27, 82, 47, 78, 91, 64, 52, 15) \\ 
  24 & (57, 46, 62, 36, 42, 26, 109, 24, 71, 58, 33, 69, 34, 37, 58, 38) \\ 
  25 & (27, 52, 47, 18, 5, 89, 111, 6, 7, 66, 17, 110, 75, 18, 55, 97) \\ 
  26 & (97, 30, 101, 29, 24, 1, 11, 0, 9, 53, 104, 43, 20, 91, 79, 108) \\ 
  27 & (74, 31, 77, 24, 21, 21, 98, 67, 67, 95, 54, 6, 19, 76, 39, 31) \\ 
  28 & (80, 83, 22, 41, 65, 10, 77, 30, 63, 57, 58, 46, 55, 57, 33, 23) \\ 
  29 & (60, 85, 64, 14, 10, 99, 57, 57, 4, 34, 35, 91, 61, 14, 61, 54) \\ 
  30 & (55, 83, 104, 37, 1, 99, 32, 21, 2, 78, 18, 27, 63, 10, 47, 123) \\ 
  31 & (26, 26, 47, 103, 55, 2, 13, 84, 49, 104, 62, 16, 80, 33, 42, 58) \\ 
  32 & (18, 116, 79, 61, 9, 41, 13, 23, 28, 72, 20, 60, 43, 66, 77, 74) \\ 
  33 & (65, 68, 61, 8, 38, 53, 52, 74, 7, 71, 0, 57, 47, 49, 82, 68) \\ 
\end{tabular}
\end{table}

\newpage
\item Table 2, Cases 34--66:

\begin{table}[htbp]
\centering
\small
\begin{tabular}{cc}
  Case & $(D_{0, 0, \ldots, 0}; D_{0, 0, \ldots, 1}; \ldots, D_{1, 1, \ldots, 0}; D_{1, 1, 
\ldots, 1})$ \\  
  \hline
  34 & (15, 61, 73, 18, 58, 9, 31, 101, 89, 78, 56, 76, 20, 56, 34, 25) \\ 
  35 & (91, 20, 10, 66, 24, 91, 0, 50, 100, 69, 55, 101, 38, 4, 44, 37) \\ 
  36 & (27, 8, 47, 71, 96, 29, 88, 23, 73, 23, 78, 13, 66, 82, 0, 76) \\ 
  37 & (66, 88, 60, 24, 105, 8, 0, 2, 59, 2, 74, 69, 68, 55, 21, 99) \\ 
  38 & (1, 67, 10, 74, 75, 55, 85, 63, 20, 55, 54, 20, 45, 80, 65, 31) \\ 
  39 & (70, 3, 63, 45, 110, 36, 36, 32, 62, 2, 36, 17, 59, 77, 69, 83) \\ 
  40 & (69, 6, 62, 25, 43, 58, 42, 73, 33, 64, 40, 57, 39, 52, 65, 72) \\ 
  41 & (81, 64, 16, 13, 78, 66, 55, 58, 63, 57, 28, 83, 33, 27, 46, 32) \\ 
  42 & (91, 28, 15, 0, 106, 75, 2, 113, 75, 70, 8, 18, 57, 92, 48, 2) \\ 
  43 & (46, 57, 51, 99, 97, 108, 5, 55, 5, 25, 43, 21, 81, 36, 17, 54) \\ 
  44 & (28, 29, 25, 83, 77, 52, 86, 75, 78, 33, 43, 4, 16, 33, 62, 76) \\ 
  45 & (39, 74, 55, 44, 0, 4, 24, 97, 60, 6, 70, 37, 64, 38, 80, 108) \\ 
  46 & (48, 23, 120, 64, 24, 17, 26, 101, 9, 34, 134, 6, 112, 35, 9, 38) \\ 
  47 & (69, 67, 35, 29, 87, 44, 75, 49, 30, 15, 12, 89, 56, 14, 69, 60) \\ 
  48 & (96, 17, 33, 34, 47, 66, 73, 40, 14, 71, 78, 35, 99, 4, 82, 11) \\ 
  49 & (39, 14, 28, 11, 64, 67, 37, 53, 85, 55, 62, 53, 78, 30, 98, 26) \\ 
  50 & (61, 38, 73, 78, 27, 40, 24, 78, 21, 61, 67, 45, 59, 27, 31, 70) \\ 
  51 & (50, 56, 39, 66, 97, 25, 96, 24, 46, 38, 12, 12, 79, 66, 83, 11) \\ 
  52 & (98, 7, 24, 66, 22, 57, 0, 20, 51, 116, 27, 38, 74, 95, 75, 30) \\ 
  53 & (115, 38, 15, 53, 62, 94, 30, 55, 37, 62, 30, 8, 84, 76, 0, 41) \\ 
  54 & (45, 81, 75, 70, 41, 39, 71, 35, 21, 9, 11, 78, 71, 1, 87, 65) \\ 
  55 & (11, 57, 41, 74, 86, 45, 81, 89, 66, 59, 63, 18, 0, 62, 12, 36) \\ 
  56 & (16, 68, 93, 62, 67, 55, 92, 23, 31, 88, 4, 1, 89, 17, 56, 38) \\ 
  57 & (72, 50, 94, 24, 101, 101, 26, 78, 50, 32, 48, 44, 17, 42, 15, 6) \\ 
  58 & (8, 50, 2, 88, 99, 48, 23, 93, 39, 77, 66, 102, 56, 11, 32, 6) \\ 
  59 & (53, 39, 23, 68, 46, 69, 77, 90, 10, 68, 44, 61, 17, 46, 39, 50) \\ 
  60 & (22, 28, 51, 86, 63, 76, 49, 43, 55, 88, 44, 10, 5, 84, 65, 31) \\ 
  61 & (66, 71, 15, 109, 36, 34, 111, 33, 21, 9, 22, 30, 52, 52, 79, 60) \\ 
  62 & (69, 53, 102, 43, 44, 8, 90, 4, 18, 90, 0, 85, 10, 42, 73, 69) \\ 
  63 & (39, 74, 94, 66, 61, 56, 30, 47, 37, 3, 25, 48, 17, 117, 59, 27) \\ 
  64 & (91, 42, 75, 85, 1, 79, 11, 28, 99, 20, 73, 20, 20, 2, 106, 48) \\ 
  65 & (46, 97, 55, 99, 67, 87, 4, 79, 7, 75, 7, 57, 21, 25, 3, 71) \\ 
  66 & (46, 56, 47, 32, 10, 15, 68, 36, 85, 39, 25, 62, 61, 106, 18, 94) \\ 
\end{tabular}
\end{table}

\newpage
\item Table 3, Cases 67--100:

\begin{table}[htbp]
\small
\centering
\begin{tabular}{cc}
 Case & $(D_{0, 0, \ldots, 0}; D_{0, 0, \ldots, 1}; \ldots, D_{1, 1, \ldots, 0}; D_{1, 1, 
\ldots, 1})$ \\  
  \hline
  67 & (67, 23, 32, 45, 71, 18, 85, 75, 37, 21, 2, 65, 102, 35, 45, 77) \\ 
  68 & (76, 53, 65, 3, 53, 116, 72, 40, 7, 32, 9, 21, 40, 84, 65, 64) \\ 
  69 & (37, 48, 41, 68, 69, 59, 41, 71, 61, 44, 22, 58, 71, 37, 49, 24) \\ 
  70 & (63, 75, 69, 5, 72, 61, 25, 68, 75, 56, 23, 9, 92, 22, 41, 44) \\ 
  71 & (64, 18, 60, 78, 20, 21, 51, 112, 7, 72, 51, 39, 51, 63, 23, 70) \\ 
  72 & (44, 43, 30, 62, 81, 112, 43, 75, 56, 3, 6, 43, 91, 68, 0, 43) \\ 
  73 & (67, 3, 12, 11, 38, 10, 83, 72, 84, 49, 63, 83, 75, 74, 13, 63) \\ 
  74 & (44, 23, 52, 28, 7, 18, 77, 82, 59, 76, 94, 58, 74, 25, 53, 30) \\ 
  75 & (13, 20, 52, 3, 64, 21, 53, 63, 35, 53, 31, 73, 64, 77, 89, 89) \\ 
  76 & (68, 24, 22, 13, 87, 68, 20, 59, 78, 13, 50, 98, 59, 37, 29, 75) \\ 
  77 & (60, 35, 73, 59, 25, 11, 91, 20, 43, 6, 103, 6, 89, 59, 26, 94) \\ 
  78 & (47, 70, 84, 18, 19, 62, 69, 30, 46, 33, 72, 70, 71, 14, 62, 33) \\ 
  79 & (83, 32, 64, 64, 24, 22, 14, 58, 51, 51, 68, 50, 66, 68, 65, 20) \\ 
  80 & (14, 78, 75, 2, 52, 54, 12, 65, 32, 34, 51, 84, 59, 41, 79, 68) \\ 
  81 & (63, 77, 51, 59, 97, 40, 34, 102, 78, 102, 2, 8, 15, 23, 20, 29) \\ 
  82 & (48, 14, 45, 64, 65, 39, 55, 76, 90, 72, 48, 50, 62, 46, 6, 20) \\ 
  83 & (85, 27, 63, 76, 41, 71, 54, 60, 10, 40, 18, 67, 26, 72, 60, 30) \\ 
  84 & (66, 22, 96, 31, 76, 5, 51, 51, 28, 26, 30, 93, 66, 93, 14, 52) \\ 
  85 & (10, 93, 45, 4, 86, 50, 63, 65, 77, 80, 59, 32, 38, 12, 8, 78) \\ 
  86 & (68, 75, 26, 18, 4, 47, 70, 43, 91, 76, 98, 18, 57, 2, 37, 70) \\ 
  87 & (49, 77, 73, 80, 69, 78, 2, 1, 62, 42, 26, 71, 2, 80, 23, 65) \\ 
  88 & (84, 87, 39, 7, 74, 13, 104, 28, 51, 28, 24, 56, 65, 75, 32, 33) \\ 
  89 & (68, 39, 90, 44, 87, 8, 63, 62, 5, 82, 82, 8, 85, 11, 25, 41) \\ 
  90 & (64, 23, 75, 52, 15, 95, 33, 80, 79, 57, 70, 2, 13, 52, 21, 69) \\ 
  91 & (85, 96, 6, 34, 17, 1, 5, 9, 77, 101, 65, 71, 55, 60, 71, 47) \\ 
  92 & (22, 58, 58, 89, 76, 74, 92, 72, 37, 20, 12, 24, 87, 9, 32, 38) \\ 
  93 & (74, 100, 7, 87, 19, 58, 33, 48, 13, 86, 37, 7, 55, 65, 74, 37) \\ 
  94 & (85, 54, 29, 70, 83, 50, 47, 55, 67, 48, 21, 0, 17, 70, 63, 41) \\ 
  95 & (31, 47, 6, 13, 62, 70, 77, 91, 60, 59, 0, 87, 23, 93, 58, 23) \\ 
  96 & (7, 37, 61, 21, 84, 93, 79, 56, 91, 58, 1, 45, 22, 74, 16, 55) \\ 
  97 & (43, 34, 60, 36, 72, 21, 38, 46, 71, 45, 34, 38, 69, 61, 56, 76) \\ 
  98 & (26, 90, 155, 16, 78, 34, 0, 33, 30, 76, 6, 58, 113, 25, 35, 25) \\ 
  99 & (40, 8, 79, 10, 71, 42, 33, 12, 31, 78, 89, 61, 60, 59, 58, 69) \\ 
  100 & (17, 27, 84, 54, 95, 13, 54, 8, 32, 68, 53, 32, 19, 96, 56, 92) \\ 
\end{tabular}
\end{table}

\end{itemize}

\subsection{Additional simulation studies}
\label{subsec:add-simu}

To examine the performance of the Bayesian interval under more settings, we extend Section \ref{sec:simu} and consider a imbalanced $2^2$ factorial design with 
$
(n_1, n_2, n_3, n_4) = (150, 150, 250, 250).
$
Again we focus on
$
\bar \tau_1,
$
and repeat the same practice in Section \ref{sec:simu} by generating and analyzing 100 simulation cases. For brevity we omit reporting the simulation cases. 

We plot the coverage rates in Figure \ref{fg:simu2}. Again, the Neymanian interval generally over-covers $\bar \tau_1.$ The independent Bayesian interval manages to correct the over-coverage of the Neymanian interval and sometimes slightly under-covers, implying the need of sensitivity analysis. 

\begin{figure}[htbp]
\centering
\includegraphics[width = .8\linewidth]{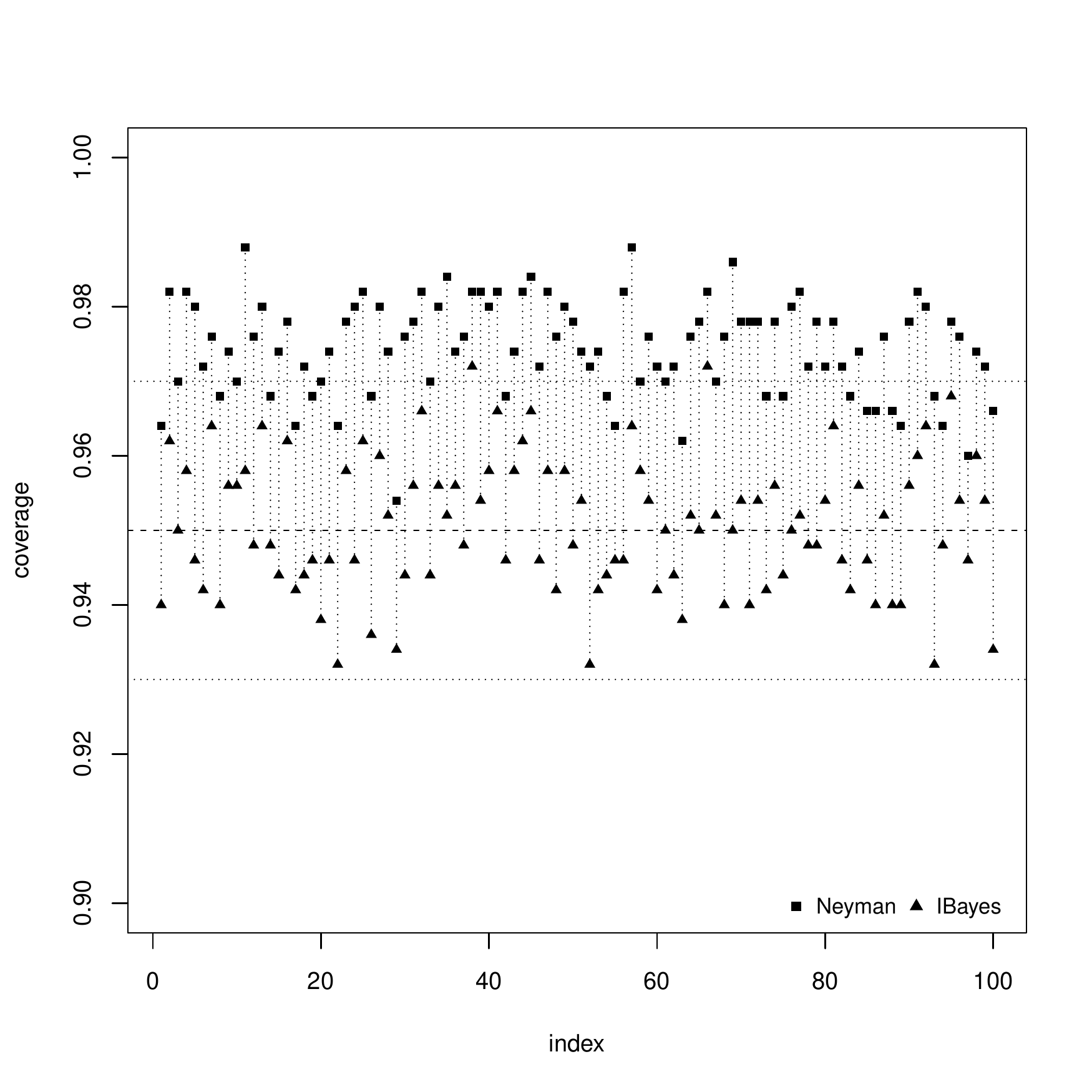}
\caption{Additional simulation results: The horizontal axis represents the case index, and the vertical shows the coverage rates for the 95\% Neymanian (rectangular) and independent Bayesian (triangular) intervals.}
\label{fg:simu2}
\end{figure}

\end{document}